\documentclass[letterpaper, 10 pt, conference]{ieeeconf}  
\pdfoutput=1

\makeatletter
\let\NAT@parse\undefined
\makeatother

\IEEEoverridecommandlockouts
\usepackage[colorlinks=true, linkcolor=blue, citecolor=red, urlcolor=magenta]{hyperref}

\usepackage{amsmath, amssymb, amsfonts, amsthm}

\theoremstyle{plain}% default
\newtheorem{thm}{Theorem}
\newtheorem{lem}{Lemma}

\theoremstyle{remark}
\newtheorem{rem}{Remark}

\theoremstyle{definition}
\newtheorem{definition}{Definition}
\newtheorem{assumption}{Assumption}

\DeclareMathAlphabet{\pazocal}{OMS}{zplm}{m}{n}

\DeclareMathOperator*{\argmin}{argmin}

\usepackage{cite}

\usepackage{algorithmic}
\usepackage[ruled,linesnumbered]{algorithm2e}
\SetKwComment{Comment}{/* }{ */}

\usepackage{graphicx}
\usepackage{textcomp}
\usepackage{mathptmx} % Times font
\usepackage{siunitx}
\usepackage{multirow}
\usepackage{array}
\usepackage{bm}
\usepackage{float}
\usepackage{wrapfig}

\usepackage{pgf,tikz}
\usetikzlibrary{arrows, shapes}

\newcommand\copyrighttext{%

	\footnotesize \textcopyright 2025 IEEE. Personal use of this material is permitted.  Permission from IEEE must be obtained for all other uses, in any current or future media, including reprinting/republishing this material for advertising or promotional purposes, creating new collective works, for resale or redistribution to servers or lists, or reuse of any copyrighted component of this work in other works.}

\newcommand\copyrightnotice{%

	\begin{tikzpicture}[remember picture,overlay]

		\node[anchor=south,yshift=5pt] at (current page.south) {\fbox{\parbox{\dimexpr\textwidth-\fboxsep-\fboxrule\relax}{\copyrighttext}}};

	\end{tikzpicture}%

}

\title{\LARGE \bf A Set-Theoretic Robust Control Approach for Linear Quadratic Games with Unknown Counterparts}

\author{Francesco Bianchin$^{1}$, Robert Lefringhausen$^{1}$, Elisa Gaetan$^{2}$, Samuel Tesfazgi$^{1}$, Sandra Hirche$^{1}$%
	\thanks{$^{1}$School of Computation, Information and Technology, Technical University of Munich, Germany
		{\tt\small \{francesco.bianchin, robert.lefringhausen, samuel.tesfazgi, sandra.hirche\}@tum.de} }
	\thanks{$^{2}$Engineering Department ’Enzo Ferrari’, University of Modena and Reggio Emilia, Modena, Italy. 
		{\tt\small elisa.gaetan@unimore.it} }
}

\begin{document}
\maketitle
\copyrightnotice
\thispagestyle{empty}
\pagestyle{empty}

%%%%%%%%%%%%%%%%%%%%%%%%%%%%%%%%%%%%%%%%%%%%%%%%%%%%%%%%%%%%%%%%%%%%%%%%%%%%%%%%
\begin{abstract}
Ensuring robust decision-making in multi-agent systems is challenging when agents have distinct, possibly conflicting objectives and lack full knowledge of each other’s strategies. This is clear in safety-critical applications such as human-robot interaction and assisted driving, where uncertainty arises not only from unknown adversary strategies but also from external disturbances. To address this, the paper proposes a robust adaptive control approach based on linear quadratic differential games. The method allows a controlled agent to iteratively refine its belief about the adversary’s strategy and disturbances using a set-membership approach, while simultaneously adapting its policy to guarantee robustness against the uncertain adversary policy and improve performance over time. We formally derive theoretical guarantees on the robustness of the proposed control scheme and its convergence to $\epsilon$-Nash strategies. The effectiveness of our approach is demonstrated in numerical simulations.
\end{abstract}

%%%%%%%%%%%%%%%%%%%%%%%%%%%%%%%%%%%%%%%%%%%%%%%%%%%%%%%%%%%%%%%%%%%%%%%%%%%%%%%%
\vspace{0cm}
\section{INTRODUCTION}
Many real-world scenarios involve multiple agents with distinct objectives interacting in a shared system. Game theory provides a structured framework for analyzing these strategic interactions over time \cite{isaacs1965differential, Starr1969}, with applications in human-robot interaction \cite{MUSIC202010216, Li_HRI} and assisted driving \cite{JAS-2019-0028}. In general, interactions can be competitive, cooperative, or a mix of both. Non-cooperative games offer a flexible framework for decision-making when objectives are neither fully aligned nor strictly opposed \cite{basar_book}.
A fundamental concept in this setting is the Nash equilibrium, a stable strategy profile where no agent benefits from unilaterally deviating, making it essential for understanding strategic interactions in dynamic and competitive environments.

Classical game formulations typically assume full information, where agents have complete knowledge of each other’s strategies and objectives. While this allows direct equilibrium computation, it is often unrealistic in practice. For instance, in human-robot interaction, human behavior is difficult to model beforehand, requiring adaptation to observed behavior rather than reliance on predefined strategies \cite{li2019differential}.
In safety-critical applications, unknown or evolving opponent strategies can significantly impact the system's ability to maintain reliable performance. To address this challenge, the first step is to integrate learning and adaptation by refining the controlled agent's belief about the adversary over time and leveraging it to improve control performance. This must be done while managing uncertainty and ensuring convergence to stable, effective strategies, without requiring full equilibrium computation.

Recent works have incorporated online learning into differential games, enabling agents to converge to equilibrium strategies in real time while adapting to dynamic environments. For instance, multi-agent reinforcement learning \cite{zhang2021multiagentreinforcementlearningselective} leverages approximate dynamic programming to iteratively refine strategies, thereby achieving convergence to Nash equilibria \cite{VAMVOUDAKIS20111556}. Linear-quadratic (LQ) games are particularly relevant in this setting due to their strong analytical tractability, which allows for explicit equilibrium solutions and stability guarantees \cite{engwerda}. Building on these properties, several iterative schemes have been proposed to ensure convergence to a Nash equilibrium. In continuous time, algorithms based on the iterative solution of coupled algebraic Riccati equations have been developed to support online control schemes that converge to equilibrium strategies \cite{li1995lyapunov,engwerda2007algorithms}, effectively handling coupled dynamics and cost interactions. In discrete time, related approaches employ Lyapunov or Riccati updates to achieve convergence \cite{yang2019data,LQ_games_learning}.

While effective from a theoretical perspective, the methods presented in \cite{zhang2021multiagentreinforcementlearningselective,VAMVOUDAKIS20111556,engwerda,li1995lyapunov,engwerda2007algorithms,yang2019data,LQ_games_learning} rely on restrictive assumptions when applied to interactive scenarios. In particular, they typically require players to alternate between learning and control phases \cite{LQ_games_learning}, meaning that each player updates its policy while the others keep theirs fixed to allow for identification. This explicit separation between learning and control is unrealistic in dynamic environments, where agents must continuously adapt in real time. Other approaches, such as \cite{li2019differential,VAMVOUDAKIS20111556,Zhang_online_ADP}, instead derive adaptive laws that remove the need for such separation. However, these schemes guarantee convergence only under the assumption that all players follow the same adaptation law. Consequently, they do not fully address the fundamental challenge of learning and adapting in the presence of an unknown adversary, especially when no strong assumptions can be made about its adaptation strategy. In addition, existing methods often neglect other sources of uncertainty, such as exogenous disturbances, which can significantly affect system performance and the overall robustness of the resulting control policies.

This paper proposes a learning-based control framework for LQ games that guarantees robust performance against both unknown adversary strategies and disturbances. Using set-membership estimation, the set of plausible adversary policies is iteratively refined from observed data. This contrasts with explicit inverse game theory \cite{ASL2024105936}, which seeks to directly parameterize the players’ cost functions. Robust LQR design via linear matrix inequalities (LMIs) then ensures closed-loop stability for all unfalsified strategies and convergence to a neighborhood-optimal solution, formalized as an $\epsilon$-Nash equilibrium. The framework thus combines the tractability of LQ games with the robustness of set-membership estimation and LMI-based design. A numerical simulation demonstrates robust stability and near-equilibrium performance even under arbitrary adversary adaptation and disturbances.

The paper is organized as follows. Section II introduces the game-theoretic problem, and Section III presents the set-theoretic approach for identifying the adversary's strategy. Section IV describes the robust policy update scheme, while Section V analyzes convergence. Section VI reports simulation results, and Section VII offers concluding remarks.

\vspace{0cm}
\section{Problem Formulation}
Consider a linear-quadratic differential game governed by system dynamics\footnote{\textbf{Notation:} Bold lowercase/uppercase symbols represent vectors/matrices, respectively. \(\mathbb{R}\) denotes the set of real numbers. Given a matrix
$\bm{A}$, its vectorization is denoted by $\operatorname{vec}(\bm{A})$, its transpose by
$\bm{A}^T$, and its induced 2-norm by $||\bm{A}||$. Given a square matrix $\bm{B}$, $\bm{B} \succ 0$ ($\bm{B} \succeq 0$) denotes that $\bm{B}$ is positive definite (positive semidefinite), and $\operatorname{Tr}(\bm{B})$ denotes its trace. $\otimes$ indicates the Kronecker product. The $n \times n$ identity matrix is indicated by $\bm{I}_n$.}
\begin{align} \label{eq:sys_dyn}
    \dot{\bm{x}}(t) = \bm{A}\bm{x}(t) + \bm{B}_1\bm{u}_1(t) + \bm{B}_2\bm{u}_2(t) + \bm{w}(t),
\end{align}
where $\bm{x} \in \mathbb{R}^{n_x}$ is the system state, $\bm{u}_1 \in \mathbb{R}^{n_{u_1}}$ is the control input of the agent under control, while $\bm{u}_2 \in \mathbb{R}^{n_{u_2}}$ is the control input of an adversary agent whose strategy is unknown, and $\bm{w}(t) \in \mathbb{R}^{n_x}$ represents a general exogenous disturbance that is bounded in magnitude. The matrices $\bm{A} \in \mathbb{R}^{n_x \times n_x}$, $\bm{B}_1 \in \mathbb{R}^{n_x \times n_{u_1}}$, and $\bm{B}_2 \in \mathbb{R}^{n_x \times n_{u_2}}$ characterize the system dynamics, with $\bm{A}$ and $\bm{B}_1$ being known to the controlled agent, while $\bm{B}_2$ is unknown, along with the control strategy of the adversary player $\bm{u}_2$.

Each player aims to minimize an infinite-horizon quadratic cost functional of the form
\begin{align} \label{eq:hum-cost}
    J_i(\bm{x}(\cdot),\bm{u}_i(\cdot)) = \int_{t_0}^{\infty} \left( \bm{x}(t)^T \bm{Q}_i \bm{x}(t) + \bm{u}_i(t)^T \bm{R}_i \bm{u}_i(t) \right) dt \notag \\ i \in \{1,2\},
\end{align}
where $ \bm{Q}_i \succeq 0 $ is the state-weighting matrix and $ \bm{R}_i \succ 0 $ is the input-weighting matrix, penalizing both state deviations and control efforts. For the controlled agent, the weighting matrices $ \bm{Q}_2 $ and $ \bm{R}_2 $ of the adversary are unknown parameters. We now introduce the concept of Nash equilibrium.
\begin{definition} [\cite{engwerda}]
    A set of strategies $\{\bm{u}^*_1,\bm{u}^*_2\} \in \Gamma$, where $\Gamma$ represents the admissible strategy space, represents a Nash equilibrium if
    \begin{equation} \label{eq:nash}
    J_i(\bm{x}, \bm{u}_i, \bm{u}_j) \leq J_i(\bm{x}, \bm{u}_i, \bm{u}_j^*) \quad \forall \bm{u}_i \in \Gamma,
\end{equation}
\end{definition}
We consider the scenario where players have access to the state and can implement feedback strategies. In the context of linear quadratic feedback games over an infinite horizon, optimal linear feedback laws exist \cite{engwerda}, allowing the strategy set to be restricted to

\begin{equation}
    \Gamma_{k,i} = \{ \bm{u}_i(\cdot) \; | \; \bm{u}_i(\cdot) = -\bm{K}_i \bm{x}(\cdot), \bm{K} \in \mathbb{R}^{n_{u_i} \times n_x} \}.
\end{equation}
        
In addition, we assume that a stabilizing Nash equilibrium exists, which can be formalized as follows \cite{li1995lyapunov}:
\begin{assumption}
    At least one of the triples $(\bm{A}, \bm{B}_1 , \sqrt{\bm{Q}_1})$, $(\bm{A}, \bm{B}_2, \sqrt{\bm{Q}_2})$ is stabilizable-detectable.
\end{assumption}
This condition is natural, as it ensures that at least one control agent has the ability to influence and observe the system's unstable modes. The above assumption also guarantees the existence of stabilizing solutions to the coupled algebraic Riccati equations (CAREs)
\begin{align}
\begin{split}
    \bm{A}^T \bm{P}_i + \bm{P}_i \bm{A} - \bm{P}_i \bm{B}_i \bm{R}_i^{-1} \bm{B}_i^T \bm{P}_i - \bm{P}_i \bm{B}_j \bm{R}_j^{-1} \bm{B}_j^T \bm{P}_j + \bm{Q}_i = 0, \quad \\ i \neq j, \quad i, j \in \{1,2\},
\end{split}
\end{align} such that feedback matrix $\bm{K}_i = \bm{R}_i^{-1} \bm{B}_i^T \bm{P}_i$ leads to a Nash equilibrium solution \cite{engwerda}. \\
We consider the problem where the adversary agent's input can be expressed as  
\begin{align}  
    \label{full_input}  
    \bm{u}_2(t) = - \bm{K}_2 \bm{x}(t) + \tilde{\bm{u}}_2(t), 
\end{align}  
where \(-\bm{K}_2 \bm{x}(t) \in \mathbb{R}^{n_{u_2}}\) represents the equilibrium strategy of the adversary, which is unknown to the controlled agent, and \(\tilde{\bm{u}}_2(t) \in \mathbb{R}^{n_{u_2}}\) accounts for deviations from this equilibrium, reflecting the adaptation dynamics as the adversary approaches the equilibrium.
The lumped disturbance \(\tilde{\bm{w}}(t) = \bm{w}(t) + \bm{B}_2 \tilde{\bm{u}}_2(t) \in \mathbb{R}^{n_x}\) captures uncertainty from external perturbations and the non-stationary behavior of the adversary. We impose no assumptions on the stochastic properties of \(\tilde{\bm{w}}(t)\) beyond its membership in a compact set \(\pazocal{W}\), defined as a convex polytope
\begin{equation}
    \label{dist_set}
    \pazocal{W} = \{\tilde{\bm{w}} \in \mathbb{R}^{n_x} \mid \bm{G}_w \tilde{\bm{w}} \leq \bm{g}_w\}.
\end{equation}
The main objective is to determine a feedback gain $\bm{K}_1$ for the controlled agent that minimizes the cost function while ensuring robustness to uncertain adversary dynamics, as we guarantee system stability under all consistent adversary strategies. This involves estimating the adversary's policy while accounting for bounded disturbances, exploiting available samples $\{\dot{\bm{x}}_{k},\bm{x}_{k},\bm{u}_{1,k}\}$, where $k \in \{1,\dots,N\}$ is the sample index, and performing online policy adaptation based on the acquired knowledge. The goal is to achieve an approximately optimal policy that is robust, as it ensures stability and resilience to adversary deviations.

\vspace{0cm}
\section{A Set Membership Approach to Adversary Strategy Identification}
\label{set-mem}
In this section, we focus on how to identify the adversary's strategy by characterizing a set $\Omega$, which describes the adversary’s policy $\bm{B}_2 \bm{K}_2$. This analysis is conducted under the assumption that the disturbance terms $\bm{w}(t)$ and $\tilde{\bm{u}}_2(t)$ are bounded. Since the adversary's input matrix \(\bm{B}_2\) is unknown, it is included in the estimation process. This set-theoretic description is then refined using a data-driven approach, where trajectory samples help discard falsified adversary models. To this end, we rewrite the game dynamics using (\ref{full_input})
\begin{align}
    \label{rev_dyn}
    \begin{split}
    \dot{\bm{x}}(t) &= \bm{A} \bm{x}(t) + \bm{B}_1 \bm{u}_1(t) - \bm{B}_2 \bm{K}_2 \bm{x}(t) + \bm{B}_2 \tilde{\bm{u}}_2(t)  + \bm{w}(t) \\
     &= \bm{A} \bm{x}(t) + \bm{B}_1 \bm{u}_1(t) - \bm{B}_2 \bm{K}_2 \bm{x}(t) + \tilde{\bm{w}}(t).
     \end{split}
\end{align}
Instead of performing a best-fit identification of the adversary's strategy, as in \cite{LQ_games_learning}, we consider the set of adversary models consistent with the recorded input-state data, following identification approaches similar to those in \cite{Bisoffi2020ControllerDF} and \cite{kerz2024safe} for safe online stochastic control. Our goal is to provide an uncertainty-aware description of the adversary's strategy \(\bm{B}_2 \bm{K}_2\), in order to enable robust responses.
By sampling the system at time instants \( k \) and collecting data samples \( \{\dot{\bm{x}}_{k},\bm{x}_{k},\bm{u}_{1,k}\} \), the set of unfalsified models is defined as  
\begin{align}
\label{gen_model}
\begin{split}
    {\Omega} := \{ &\bm{B}_2\bm{K}_2 \in \mathbb{R}^{n_x \times n_x } \; | \\ 
    & \dot{\bm{x}}_{k} - \bm{A} \bm{x}_{k} - \bm{B}_1 \bm{u}_{1,k} + \bm{B}_2 \bm{K}_2 \bm{x}_{k} \in {\pazocal{W}}, \\ &\forall k \in \{1,\dots,N\} \}.
\end{split}
\end{align}  
% which is well behaved, assuming that the unknown disturbances $\{\bm{w}_{\tau}\}$ don't leave $\mathcal{W}$. \\
In the following, we assume that $\Omega$ is bounded, which is not restrictive in practice. This condition can be verified as follows.
\begin{lem}[\cite{Bisoffi2020ControllerDF}] \label{lem_set} 
The set of consistent models \(\Omega\) is convex and closed. It is bounded if and only if the data generating the set satisfies \(\text{rank}\begin{bmatrix} \bm{x}_{0} \cdots \bm{x}_{k} \cdots \bm{x}_{N} \end{bmatrix} = n_x\), where \(N\) is the number of available samples and \(n_x\) is the state-space dimension.
\end{lem}
The condition is easily satisfied, as it requires a low amount of linearly independent state samples.
Given the noise set description ${\pazocal{W}}$, the adversary strategy set ${\Omega}$ is explicitly derived by reformulating the inequality constraints on the disturbance into constraints on the adversary strategy. In particular, using the dynamics relation (\ref{rev_dyn}), inequality (\ref{dist_set}) is rewritten for each sample as
\begin{subequations}
\begin{align}
    \bm{G}_w \left( \dot{\bm{x}}_{k} - \bm{A} \bm{x}_{k} - \bm{B}_1 \bm{u}_{1,{k}} + \bm{B}_2 \bm{K}_2 \bm{x}_{k} \right) \leq \bm{g}_w, \\
    \bm{G}_w \bm{B}_2 \bm{K}_2 \bm{x}_{k} \leq \bm{g}_w - \bm{G}_w \left( \dot{\bm{x}}_{k} - \bm{A} \bm{x}_{k} - \bm{B}_1 \bm{u}_{1,{k}} \right).
\end{align}
\end{subequations}
We now define  
\begin{align}
    \bm{b}_{k} = \bm{g}_w - \bm{G}_w \left( \dot{\bm{x}}_{k} - \bm{A} \bm{x}_{k} - \bm{B}_1 \bm{u}_{1,{k}} \right),
\end{align}
and recall the vectorization property  
\begin{align}
    \operatorname{vec}(\bm{G}_w \bm{B}_2 \bm{K}_2 \bm{x}_{k}) = \left( \bm{x}_{k}^T \otimes \bm{G}_w \right) \operatorname{vec}(\bm{B}_2 \bm{K}_2).
\end{align}
We additionally define  
\begin{align}
    \bm{E}_{k} = \left( \bm{x}_{k}^T \otimes \bm{G}_w \right),
\end{align}  
in order to derive the vectorized form of the inequality constraints for the set of adversary strategy parameters \( \bm{B}_2 \bm{K}_2 \)
\begin{align}
    \label{data-constr}
    \bm{E}_{k} \operatorname{vec}(\bm{B}_2 \bm{K}_2) \leq \bm{b}_{k}.
\end{align}
Accounting for all available samples, the following representation of the adversary strategy set is derived
\begin{align}
\label{iter_poly}
\begin{split}
     \Omega = \{ &\bm{B}_2\bm{K}_2 \; | \; \bm{E}_{k} \operatorname{vec}(\bm{B}_2 \bm{K}_2) \leq \bm{b}_{k} 
    \\ &\forall k \in \{1,\dots,N\} \}.
\end{split}
\end{align}
This description of \( \Omega \) is called the \( \pazocal{H} \)-representation of the polytope, as it involves a set of inequalities, each defining a half-plane in the parameter space. The effect of adding an informative data point is illustrated in Fig.~\ref{cut}.
\begin{rem}
    New inequalities may be redundant if they are less restrictive than existing constraints. A minimal polytope representation is preferable and can be achieved using ad hoc linear programs to verify redundancy before adding constraints to the $\pazocal{H}$-representation of $\Omega$ \cite{ziegler2012lectures}.
\end{rem}
While the \(\pazocal{H}\)-representation is useful for updating the knowledge of set \(\Omega\) based on new data, it is less convenient for exploiting the information it contains. As noted in \cite{matousek2013lectures}, a convex polytope can also be represented by its vertices, through the \(\pazocal{V}\)-representation, as any polytope can be viewed as the convex hull of these points. This vertex-based representation is often more practical, as it directly provides the extreme points of the polytope, which are essential for efficiently setting up optimization problems, including the one presented in Section IV. To obtain the \(\pazocal{V}\)-representation, we identify the intersections of sets of \(n_x^2 - 1\) half-planes that belong to \(\Omega\), where \(n_x^2\) is the number of parameters under identification (i.e., the entries of \(\bm{B}_2 \bm{K}_2\)), selecting those that lie on the boundary of the set.
\begin{figure}[t]
		\includegraphics[clip,trim=0cm 0cm 0cm 0cm,width=0.48\textwidth]
        {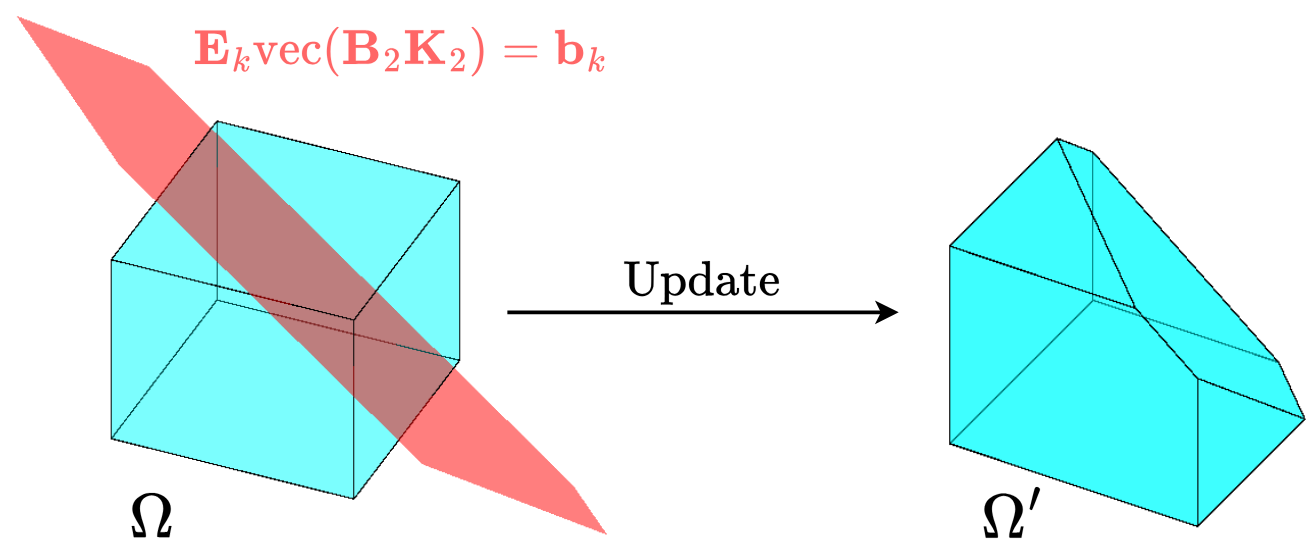}
        \caption{Illustrative example of the effect of adding a non-redundant constraint, corresponding to a new data point, to the set description of $\Omega$. The constraint introduces a plane cut (shown in red) in the parameter space according to (\ref{data-constr}), reducing the polytope’s volume (thus the parameter uncertainty) and generating new vertices.}
        \label{cut}
\end{figure}

\vspace{0.3cm}
\section{Robust Response via Data-Driven LMI-based Riccati Iterations} \vspace{0.1cm}
The previous section outlined an approach for identifying a bounded set of adversary strategies consistent with observations. Here, we present a method to ensure a robust response that guarantees stability against all unfalsified adversary strategies using set-membership techniques. LMIs provide a systematic way to impose robustness constraints, ensuring stability and performance under worst-case adversary strategies. To this end, we adopt an LMI-based solution to the LQR problem arising from (\ref{eq:hum-cost}), following the approach first introduced in \cite{feron_lmi} and later extended to robust LQR in \cite{LMI_LQR}. Specifically, we propose an iterative procedure in which the controlled agent solves a one-sided LQR problem, continuously updating its control law based on the current estimate of the adversary’s strategy.

The first step towards devising a robust LQR strategy is to reformulate it as a semi-definite program (SDP). As a first step, we introduce the additional state 
\begin{align}
    \bm{z}(t) = \begin{bmatrix}
        \bm{Q_1}^{1/2} & 0  \\ 0 & \bm{R_1}^{1/2}
    \end{bmatrix} \begin{bmatrix}
        \bm{x}(t) \\ \bm{u}(t)
    \end{bmatrix},
\end{align}
which can be interpreted as the performance output for the controlled agent in an output energy minimization problem.
The LQR problem is now to design a feedback controller $K_1$ that minimizes the $H_2$ norm of the transfer function $\pazocal{G}:\tilde{\bm{w}} \rightarrow \bm{z}$ \cite{feron_lmi}. To reformulate the problem, we first define the matrix $\bm{A}_1 = \bm{A} - \bm{B}_2 \bm{K}_2$, i.e., the state transition matrix as seen by the controlled agent, and matrix $ \bm{W}_c $, the controllability Gramian of the system. We also define auxiliary variables $ \bm{Y} = \bm{K}_1 \bm{W}_c $ and $ \bm{X} $. As shown in \cite{feron_lmi, DataDriven}, the $H_2$ norm minimization problem can then be rewritten as a semi-definite program of the form
    \begin{subequations}
    \begin{align}
        \label{main_opt}
        (\hat{\bm{W}}_c,\hat{\bm{Y}},\hat{\bm{X}}) &=  \argmin_{\bm{W}_c,\bm{Y},\bm{X}} \left[ \operatorname{Tr}(\bm{Q}_1 \bm{W}_c) + \operatorname{Tr} (\bm{X}) \right] \\
        & \text{s.t.}  \; \bm{A}_1 \bm{W}_c + \bm{W}_c \bm{A}_1^T - \bm{B}_1 \bm{Y} - \bm{Y}^T \bm{B}_1^T \!+\! \bm{I} \preceq 0  \label{stab_constr_1} \\
        & \quad \begin{bmatrix}
            \bm{X} & \bm{R}_1^{1/2} \bm{Y} \\ 
            \bm{Y}^T \bm{R}_1^{1/2} & \bm{W}_c
        \end{bmatrix} \succeq 0, \quad \bm{W}_c \succ 0. \label{stab_constr_2}
    \end{align}
    \end{subequations}
 Once solved, the optimal feedback matrix is obtained as $ \hat{\bm{K}}_1 = \hat{\bm{Y}}\hat{\bm{W}}_c^{-1} $. Note that linear constraint (\ref{stab_constr_1}) implies the quadratic stability of the closed-loop system when employing the optimal control feedback $\hat{\bm{K}}_1$. Constraints (\ref{stab_constr_2}) are instead exploited to linearize the quadratic constraint on the auxiliary variable $\bm{X}$ using the Schur decomposition.

\begin{figure}[t]
        \begin{center}
            \includegraphics[clip,trim=0cm 0cm 0cm 0cm,width=0.3\textwidth]{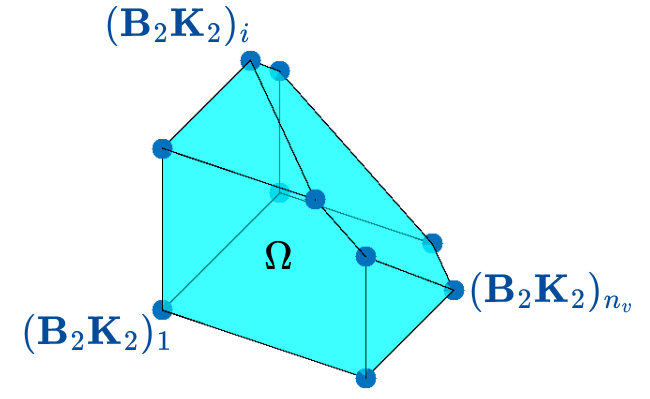}
        \end{center}
        \vspace{-0.2cm}
        \caption{The vertices of the unfalsified strategy set $\Omega$ are shown as blue markers. These vertices correspond to extremal adversary strategies, providing a finite representation of the set. By exploiting them, stability guarantees can be enforced over the entire $\Omega$.}
        \vspace{-0.2cm}
        \label{vertices}
\end{figure}

It can be seen that the formulation above accounts for a nominal description of the system, involving a unique description of $\bm{A}_1 = \bm{A}-\bm{B}_2\bm{K}_2$. In our case, we deal with a polytopic description of the adversary strategy $\bm{B}_2\bm{K}_2$ (see Fig.~\ref{vertices}), which makes the matrix $\bm{A}_1$ uncertain. Conveniently, the LMI-based formulation of LQR can be adapted for robustness by expanding the set of stability LMIs of type~\eqref{stab_constr_1} to account for the set of unfalsified adversary strategies. To guarantee robustness against all unfalsified models in ${\Omega}$, stability constraints are imposed on each vertex~\cite{LMI94} (e.g., the unfalsified strategies, shown as blue markers in Fig.~\ref{vertices}). The following theorem establishes that enforcing stability at all vertices ensures stability for the entire set.
\begin{thm}
    Consider system (\ref{rev_dyn}) and the set $\Omega$ of unfalsified adversary strategies. Assume that the data collected is such that $\Omega$ is bounded, i.e., the conditions of Lemma 1 are satisfied. Then, the satisfaction of the stability inequalities 
    \begin{align}
    \label{stab_constr}
    (\bm{A}-(\bm{B}_2\bm{K}_2)_i)\bm{W}_c+\bm{W}_c(\bm{A}-(\bm{B}_2\bm{K}_2)_i)^T-\bm{B}_1\bm{Y}-\bm{Y}^T\bm{B}_1^T \preceq 0,
\end{align}
    where the index \( i \in \{ 1, \dots, n_v \} \) denotes the vertex under analysis, 
    guarantees the stability against every unfalsified adversary strategy $\overline{\bm{B}_2\bm{K}_2} \in {\Omega}$.
\end{thm}
\begin{proof}
    First note that condition (\ref{stab_constr}) equates to imposing the quadratic stability condition for all adversary strategies being vertices of $\Omega$. For all the other elements of $\Omega$, the polytopic nature can be used to establish stability. Specifically, any element of $\Omega$ can be expressed as  
    \begin{align}
        \overline{\bm{B}_2\bm{K}_2} = \sum_{i=1}^{n_v} \lambda_i \cdot (\bm{B}_2\bm{K}_2)_i, \quad \text{with } \lambda_i \geq 0 \; \forall i, \quad \sum_{i=1}^{n_v} \lambda_i = 1.
    \end{align}
    We can now write the stability inequality for the generic element $\overline{\bm{B}_2\bm{K}_2}$
    \begin{align}
        (\bm{A}-\overline{\bm{B}_2\bm{K}_2})\bm{W}_c + \bm{W}_c(\bm{A}-\overline{\bm{B}_2\bm{K}_2})^T - \bm{B}_1\bm{Y} - \bm{Y}^T\bm{B}_1^T \preceq 0.
    \end{align}
    and note that it can be rewritten as  
    \begin{align}
        \sum_{i=1}^{n_v} \lambda_i \big[ &(\bm{A} - (\bm{B}_2\bm{K}_2)_i)\bm{W}_c 
        + \bm{W}_c(\bm{A} - (\bm{B}_2\bm{K}_2)_i)^T \notag \\ 
        & - \bm{B}_1\bm{Y} - \bm{Y}^T\bm{B}_1^T \big] \preceq 0.
    \end{align}
    We note that it is trivially satisfied since \( \lambda_i \geq 0 \) and each term in the summation satisfies the matrix inequality.
    \hfill\qedsymbol
\end{proof}
The LQR minimization problem was thus modified to enforce stability across all system configurations. Given the bounded disturbance assumption on $\pazocal{W}$, the set $\Omega$ includes all unfalsified adversary strategies, ensuring that the control law from (\ref{main_opt}) under constraints (\ref{stab_constr}) remains robustly stable. While these constraints may limit optimality with respect to the true adversary strategy, they guarantee stability across all unfalsified behaviors. Some results on the convergence to optimal behaviour, or equivalently to the Nash equilibrium strategy, are reported in Section \ref{main_proof}.
\begin{rem}
    Through our method, we perform approximate Riccati iterations, which refine the cost-to-go function at each step. These can be interpreted as value iterations, in contrast to Lyapunov iterations, which follow a policy iteration approach by separating policy evaluation and update steps, as seen in works like \cite{li1995lyapunov}.
\end{rem}

\vspace{0cm}
\textbf{Algorithmic Formulation.} The previously introduced components enable the formulation of game-theoretic robust control, where data-driven learning and policy updates are performed iteratively. In particular, the control strategy is updated based on the latest polytopic characterization of the adversary’s strategy. The complete procedure is summarized in Algorithm~\ref{main_alg}.
The initialization requires a conservative specification of the admissible adversary strategies, \(\Omega^0\), from which an initial \(\bm{K}_1^0\) is obtained via the SDP (\ref{main_opt})-(\ref{stab_constr}) (line 3). 
Following initialization, the control strategy is refined at fixed intervals \(T\). At each iteration $j$, data samples $\pazocal{D}^j$ are collected (line 7), updating the knowledge of admissible adversary strategies encoded in \(\Omega\). Specifically, data from $\pazocal{D}^j$ defines set $\Omega_{\pazocal{D}^j}$ (line 8), which is then intersected with the unfalsified models \(\Omega^{j-1}\) from the previous iteration, to obtain \(\Omega^j\).
This updated model knowledge is then used to solve the SDP-based robust LQR problem, yielding an improved feedback gain \(\bm{K}_1\) (line 10) that ensures robustness against adversarial strategies.\footnote{Code for Algorithm~\ref{main_alg} and for the simulations presented hereafter can be found at: \url{https://github.com/TUM-ITR/STAR-LQG}}

\begin{rem}
    Sampling times can be selected freely and strategically optimized to improve the identification process, e.g. leveraging techniques from experiment design in bounded disturbance settings \cite{MILANESE198678}.
\end{rem}

\begin{rem}
The use of general-form polytopes leads to a large number of constraints and vertices, which also increases the complexity of the subsequent optimization problem. Efficiency can be improved with approximations such as zonotopes or ellipsoids, which simplify vertex enumeration.
\end{rem}

\setlength{\textfloatsep}{2pt}
\RestyleAlgo{ruled}
\begin{algorithm}[t]
\caption{Game-Theoretic Robust Control}\label{main_alg}
\KwData{Learning horizon $T$ \\ Sampling time $\Delta t$ \\ Game parameters $\bm{Q}_1, \bm{R}_1$ \\ Known dynamics matrices $\bm{A}, \bm{B}_1$ \\ Lumped disturbance set $\pazocal{W}$ \\ Initial unfalsified adversary strategy set ${\Omega}^0$}
$j \gets 0$\;
$t \gets t_0$\;
$\bm{K}_1 \gets$ Solve SDP$(\bm{Q}_1, \bm{R}_1, \bm{A}, \bm{B}_1, {\Omega}^0)$ from Eq.~(\ref{main_opt}) using constraints (\ref{stab_constr}) arising from ${\Omega}^0$\;
\While{true}{
    Control system using feedback gain $\bm{K}_1$ for interval $[t,t+T]$\;
    $j \gets j + 1$\; 
    Collect samples $\pazocal{D}^j$ \;
    Given $\pazocal{W}$ and $\pazocal{D}^j$, compute ${\Omega}_{\pazocal{D}^j}$\;
    ${\Omega}^j \gets {\Omega}^{j-1} \cap {\Omega}_{\pazocal{D}^j}$\;
    $\bm{K}_1^j \gets$ Solve SDP$(\bm{Q}_1, \bm{R}_1, \bm{A}, \bm{B}_1, \bm{\Omega}^j)$ from Eq.~(\ref{main_opt}) using constraints (\ref{stab_constr}) arising from ${\Omega}^j$\;
    $\bm{K}_1 \gets \bm{K}_1^j$\;
    $t \gets t +T$\; 
}  
\end{algorithm}

\section{Convergence to an \texorpdfstring{$\epsilon$-}{epsilon-}Nash Equilibrium Strategy}
\label{main_proof}
In this section, we derive bounds on the optimality of the control law presented in Algorithm \ref{main_alg}. To achieve this, we first establish convergence results for the set-based estimation approach introduced in Section \ref{set-mem}. Understanding the convergence properties of this estimation process is crucial for ensuring the reliability of the learned adversary model and its impact on the optimality of the control strategy. However, proving convergence analytically is highly challenging, if not infeasible, due to the disturbance term \(\tilde{\bm{w}}\), which is only specified as a bounded set rather than a probabilistic model. Additionally, the complexity of high-dimensional polytopes (dimensions greater than three) prevents analytical volume formulations.  
To address this, we introduce a relaxation where the estimation polytope is approximated by outer bounding ellipsoids \cite{MILANESE1991997, FOGEL1982229}. Any convergence results under this relaxation remain valid for the original problem, as the latter imposes stricter constraints. To proceed, we assume that the disturbance components are uncoupled, so that each element of the disturbance vector can be bounded independently as \vspace{-0.1cm}
\begin{align}
    \tilde{w}_{i}^2 \leq \gamma_i^2, \quad i \in \{1,\dots,n_x\},
\end{align}
where $\tilde{w}_{i}$ denotes the $i$-th component of $\tilde{\bm{w}}$. This uncoupling assumption can always be enforced by appropriately defining the uncertainty set $\pazocal{W}$, or by extending it for the purpose of convergence analysis. 
We now reformulate the estimation model in scalar form. Starting from (\ref{gen_model}), define
\begin{align}
    \bm{y}_{k} = -\dot{\bm{x}}_{k} + \bm{A} \bm{x}_{k} + \bm{B}_1 \bm{u}_{1,k}, \qquad
    \bm{\Theta} = \bm{B}_2 \bm{K}_2 \in \Omega,
\end{align}
so that \vspace{-0.2cm}
\begin{align}
    \bm{y}_k = \bm{\Theta} \bm{x}_k + \tilde{\bm{w}}_k.
\end{align}
The scalar components are then
\begin{align}
    y_{k,i} = \bm{\theta}_i x_{k} + \tilde{w}_{k,i},
\end{align}
where $\bm{\theta}_i$ is the $i$-th row of $\bm{\Theta}$ and $k$ indexes the samples.
This formulation enables the application of results from \cite{Dasgupta}, which introduces an iterative ellipsoid-bounding procedure that incorporates new data-driven constraints to refine an outer ellipsoidal approximation
\begin{align}
    \bar{\Omega}_{k} = \left\{ \bm{\theta} : (\bm{\theta}-\bm{\theta}_{c,k})^\top \bm{S}_{k}^{-1} (\bm{\theta}-\bm{\theta}_{c,k}) \leq \sigma_{k}^2 \right\}
\end{align}
of ${\Omega}_{k}$. Here, $\bm{\theta}_{c,k}$ denotes the center of the ellipsoid at iteration $k$, and $\bm{S}_{k} \succ 0$ is the shape matrix. This leads to the following estimation convergence result.

\begin{lem} [\cite{Dasgupta}]
\label{conv_lemma}
    If there exist constants $\alpha_1 > 0$, $\alpha_2 > 0$, and $N > 0$ such that for all $t$,
    \begin{align}
    \label{PERS_EXC}
        0 < \alpha_1 \bm{I}_{n_x+1} \leq \sum_{k = t}^{t + N \Delta t} \begin{bmatrix} \bm{x}_{k} \\ \tilde{w}_{k,i} \end{bmatrix} \begin{bmatrix} \bm{x}_{k}^T \; \tilde{w}_{k,i} \end{bmatrix} \leq \alpha_2 \bm{I}_{n_x+1} < \infty,
    \end{align}
    then \vspace{-0.2cm}
    \begin{align}
        \lim_{k \rightarrow \infty} \bm{S}_{k} = \bm{S}_{\infty}, \\
        \lim_{k \rightarrow \infty} \sigma_{k}^2 \in [0,\gamma^2].
    \end{align}
\end{lem}

\setlength{\textfloatsep}{20pt plus 2pt minus 4pt}

In particular, the squared radius variable $\sigma_{k}^2$ exhibits exponential convergence. The above lemma implies that as long as the persistence of excitation condition (\ref{PERS_EXC}) is satisfied, the outer ellipsoid enclosing the true parameter set converges to a bounded region, the size of which depends on the disturbance bound. 
We now proceed to establish a theorem concerning the convergence of the proposed algorithm toward the Nash equilibrium solution.
\begin{figure*}[t]\centering
		\includegraphics[clip,trim=4.3cm 0.3cm 4.5cm 0.6cm,width=\textwidth]{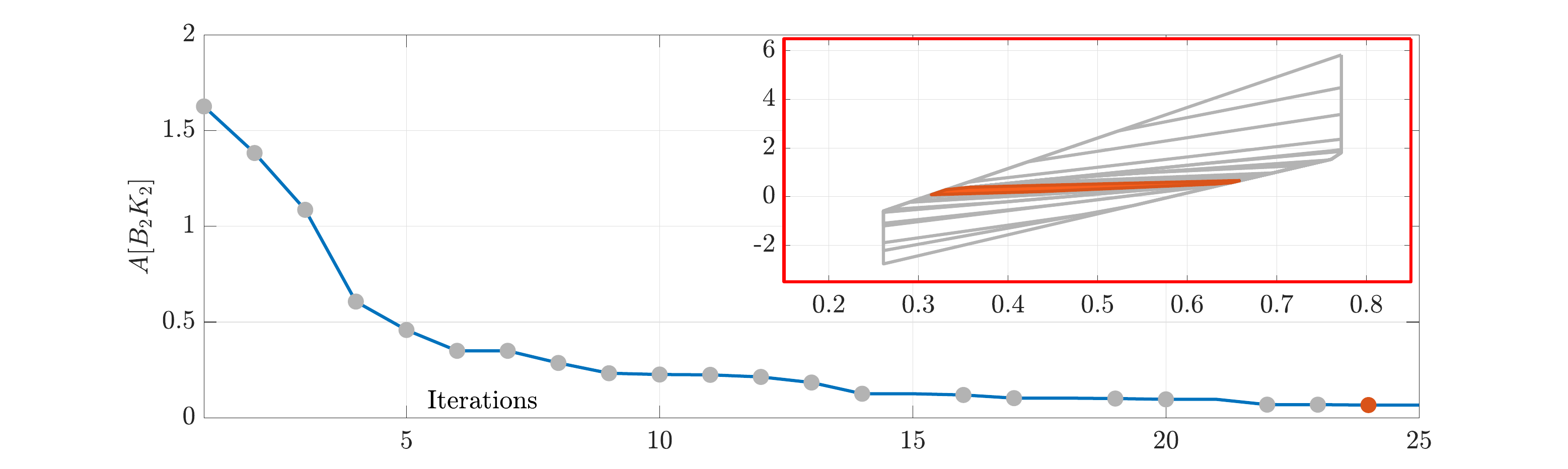}
        \vspace{-0.7cm}
        \caption{The evolution of the uncertainty area $A(\bm{B}_2 \bm{K}_2)$ for the estimated parameters $\bm{B}_2 \bm{K}_2$, when applying Algorithm 1 on system (\ref{sim_example}), is shown in blue in the main graph. This area corresponds to the polytope describing the uncertain parameters. Iterations where the uncertainty decreases are marked with gray dots, and the final settled area is shown in orange. The top-right graph illustrates the contraction of the polytopes (simple polygons) over iterations. Each dimension corresponds to one of the two estimated parameters, i.e., the non-zero entries of $\bm{B}_2 \bm{K}_2$. The red bounding box represents the initial conservative specification $\Omega^0$, gray polygons show the shrinking uncertainty sets $\Omega_k$, and the set reached at iteration 25 of Algorithm~1 is shown in orange.}
        \label{bound_est}
        \vspace{-0.3cm}
\end{figure*}
\begin{thm}
    Consider the polytope $\Omega$ defined in (\ref{iter_poly}) and its iterative update, approximated by outer ellipsoids $\bar{\Omega}_k$. Suppose the conditions of Lemma \ref{conv_lemma} hold. Then, under these conditions, the controlled agent's policy, computed according to (\ref{main_opt}), converges to an $\epsilon$-Nash strategy $\bm{\bar{u}}_1$, i.e.,
    \begin{equation}
    \bar{J}_1(\bm{x}, \bm{\bar{u}}_1, \bm{u}_2^*) \leq J_1(\bm{x}, \bm{u}_1, \bm{u}_2^*) + \epsilon, \quad \forall \bm{u}_1 \in \Gamma_1.
    \end{equation}
\end{thm}
\begin{proof}
    Under the conditions of Lemma \ref{conv_lemma}, the outer bounding ellipsoid for the uncertainty in the model parameters converges to a bounded volume around its centroid. Specifically, defining $\lambda_{\max}(\bm{S}_\infty)$ as the largest eigenvalue of the ultimate shape matrix \( \bm{S}_\infty \), we obtain an upper bound on the longest axis of the ellipsoid
    \begin{equation}
        l_{\max} \leq \gamma \sqrt{\lambda_{\max}(\bm{S}_\infty)}.
    \end{equation}
    Since this holds for all entries in the parameter matrix \( \bm{B}_2 \bm{K}_2 \), we can bound the maximum perturbation in \( \bm{B}_2 \bm{K}_2 \), or equivalently, the induced perturbation in the state transition matrix as observed by the controlled agent
    \begin{equation}
        \|\Delta \bm{A}_1\| = \|\Delta (\bm{B}_2 \bm{K}_2)\| \leq n_x \gamma \sqrt{\lambda_{\max}(\bm{S}_\infty)}.
    \end{equation}
    Using perturbation bounds for the algebraic Riccati equation (see \cite{sun}), we obtain a computable bound on the deviation of the stabilizing LQR gain
    \begin{equation} \label{K_bound}
        \|\Delta \bm{K}_1\| \leq \delta(\|\Delta \bm{A}_1\|) = \delta \left(n_x \gamma \sqrt{\lambda_{\max}(\bm{S}_\infty)}\right),
    \end{equation}
    which we will hereafter denote simply as $\delta$. Rewriting the cost-to-go function,
    \begin{equation}
        J_1 = \int_{t}^{\infty} \bm{x}(\tau)^T (\bm{Q}_1 + \bm{K}_1^T \bm{R}_1 \bm{K}_1) \bm{x}(\tau) \, d\tau,
    \end{equation}
    we express the deviation from the optimal cost:
    \begin{equation}
        {\bar{J}}_1 - J_1^* = \int_{t}^{\infty} \bm{x}(\tau)^T \left( 2 \bm{K}_1^{*T} \bm{R}_1 \Delta \bm{K}_1 + \Delta \bm{K}_1^T \bm{R}_1 \Delta \bm{K}_1 \right) \bm{x}(\tau) \, d\tau.
    \end{equation}
    Applying the bound from (\ref{K_bound}), such deviation can be further limited as
    \begin{equation}
        ({\bar{J}}_1 - J_1^*) \leq \left(2 \| \bm{K}_1^{*T} \bm{R}_1 \| \delta + \| \bm{R}_1 \| \delta^2 \right) \operatorname{Tr}(\bm{P}_1),
    \end{equation}  
    where \( \bm{P}_1 = \int_t^{\infty} \bm{x}(\tau) \bm{x}(\tau)^T d\tau \).
    Thus, setting 
    \begin{equation}
        \epsilon = \left(2 \| \bm{K}_1^{*T} \bm{R}_1 \| \delta + \| \bm{R}_1 \| \delta^2 \right) \operatorname{Tr}(\bm{P}_1),
    \end{equation}
    we conclude that the learned policy leads to an $\epsilon$-Nash equilibrium.  
    \hfill\qedsymbol
\end{proof}

\section{Simulation Results}
This simulation study serves two main purposes: (i) to demonstrate that the proposed learning-based control strategy converges to a near-Nash equilibrium as more data is collected, and (ii) to illustrate that, unlike least-square-based estimation approaches, our method guarantees robustness against all unfalsified adversary strategies. These two aspects are examined in the following subsections.
To illustrate the effectiveness of the proposed method, we present simulation results based on a practical example from \cite{LQ_games_learning}. This example involves human–robot interaction, specifically the dynamics of a contact robot described by system matrices
\begin{align}
\label{sim_example}
    \bm{A} &= 
    \begin{bmatrix}
        0 & 1 \\
        0 & \frac{D_c}{J_c}
    \end{bmatrix}, \quad
    \bm{B}_1 = 
    \begin{bmatrix}
        0 \\
        \frac{1}{J_c}
    \end{bmatrix}, \quad
    \bm{B}_2  = 
    \begin{bmatrix}
        0 \\
        \frac{b}{J_c}
    \end{bmatrix},
\end{align}
where the state vector is defined as \(\bm{x} = \begin{bmatrix} x_e - r & v_e \end{bmatrix}^T\), with \(x_e\) representing the robot's end-effector position, \(r\) the target position, and \(v_e\) the end-effector velocity. The control inputs are \(u_1\) for the robot (controlled agent) and \(u_2\) for the human agent (adversary), whose behavior is unknown. The system parameters are \(J_c = 6~\si{kg}\) (inertia), \(D_c = \text{-}~0.2~\si{N/m}\) (damping), and \(b = 0.8\), a scaling factor unknown to the robot. 

The task models human arm reaching movements, where a human agent guides the end effector from an initial to a target position with robotic assistance. This setup could be relevant in rehabilitation, when aiding patients in relearning motions, as well as in manufacturing, where operators are assisted in moving heavy objects. The state weighting matrices are \(\bm{Q}_1 = \operatorname{diag}([25, \; 0.1])\) and \(\bm{Q}_2 = \operatorname{diag}([15, \; 0.3])\), while the input weights are \(R_1 = 0.1\) and \(R_2 = 0.15\), with \(\bm{Q}_2\) and \(R_2\) unknown to the controlled agent. The exogenous disturbance is sampled uniform noise with bounds $|w_1|\leq0.5$, $|w_2| \leq 0.5$.
In this setup, the human player interacts with the system using the input  
$u_2(t) = -\bm{K}_2^* \bm{x}(t) + \tilde{u}(t)$,  
where $\bm{K}_2^* = [2.69,1.37] $ represents the equilibrium feedback strategy, while $\tilde{u}(t) = 2 \operatorname{cos}(2 \pi t) \operatorname{e}^{-0.2t}$ accounts for the non-equilibrium component as the human settles into the Nash strategy.

\subsection{Convergence Properties} 
We first focus on convergence, aiming to show that the estimation of the adversary strategy set $\Omega$ contracts around the true parameters, and that the controlled agent’s feedback policy $\hat{\bm{K}}_1$ simultaneously converges toward the Nash equilibrium strategy. We thus apply Algorithm \ref{main_alg} to steer the controlled agent. The objective is to drive the system to \(r = 0\), starting from \(\bm{x}(0) = [-3, \; 0]^T\). Estimation and policy updates occur at intervals of \(T = 0.03\)s, with data collected at fixed sampling times of \(\Delta t = 0.01\)s. Figure \ref{bound_est} illustrates the evolution of the uncertainty in the estimated parameters \(\bm{B}_2 \bm{K}_2\). The main plot shows the area enclosed by the polytope representing the uncertain parameters, which decreases over iterations. The top-right plot shows the contraction of the polytopes themselves. Overall, the estimation process converges to a polygon surrounding the true parameter values \((\bm{B}_2 \bm{K}_2)_{2*} = [0.36, \; 0.18]\), demonstrating how the uncertainty shrinks as the algorithm progresses. These updates are accompanied by policy refinements for the controlled agent’s feedback gain \(\hat{\bm{K}}_1\), as shown in Fig. \ref{control_conv}. Notably, strategy updates proceed in sync with the estimation refinements, until \(\hat{\bm{K}}_1 = [13.88, \; 12.27]\) is reached at iteration 25, closely approximating the Nash equilibrium solution \(\bm{K}_1^* = [13.81, \; 12.05]\).

\begin{figure}[t]
		\includegraphics[clip,trim=3.1cm 0.75cm 2.7cm 0.9cm,width=0.48\textwidth]{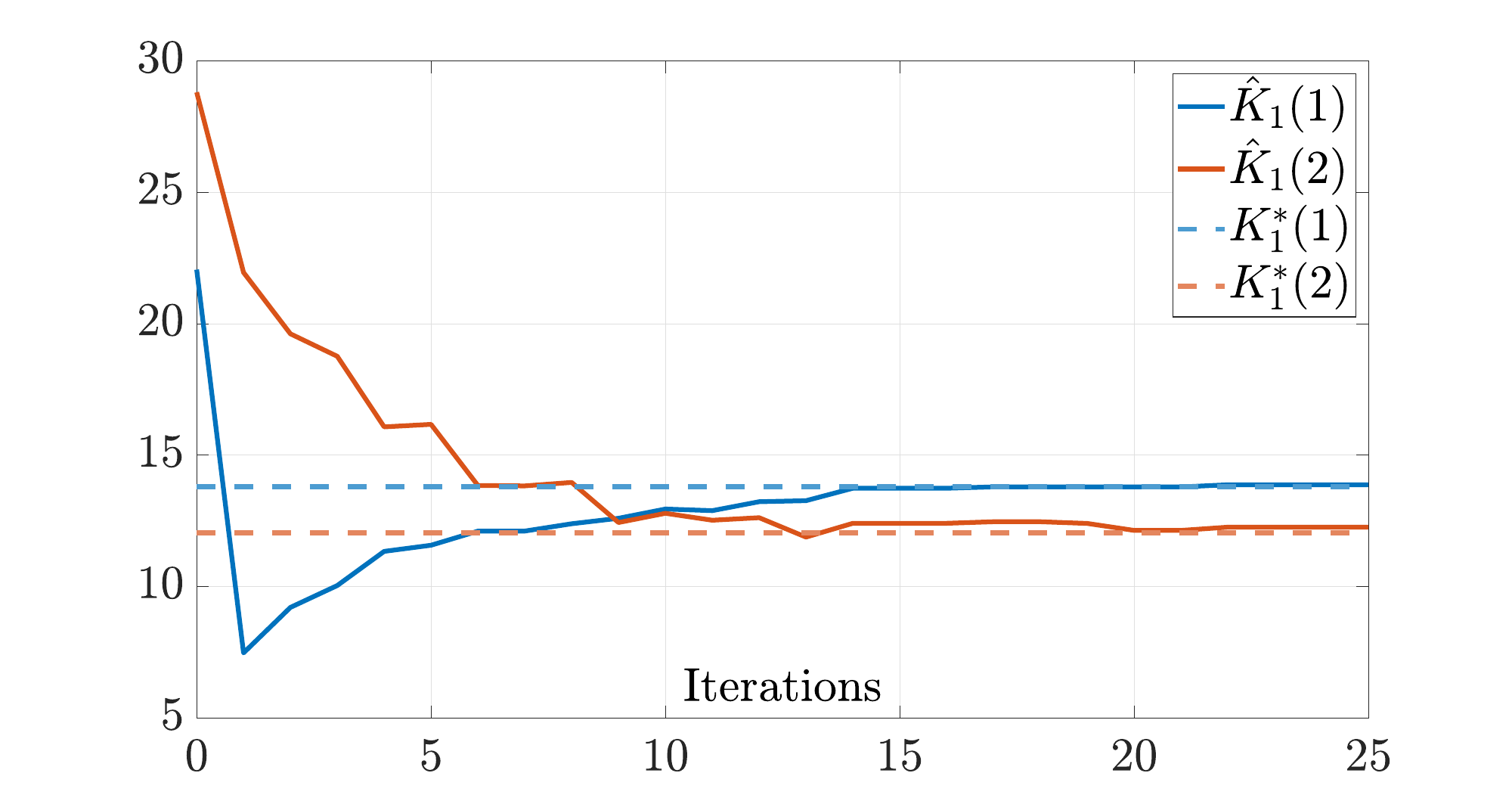}
        \vspace{-0.18cm}
        \caption{The solid lines represent the evolution of the controlled agent's feedback policy, \( \hat{\bm{K}}_1 \), across algorithm iterations, with the first entry shown in blue and the second in orange. The corresponding optimal Nash gain values are depicted by dashed lines.}
        \vspace{-0.4cm}
        \label{control_conv}
\end{figure}

\subsection{Robustness Properties} \vspace{-0.05cm}
We then turn to robustness. The following experiment demonstrates that the proposed approach preserves stability against all unfalsified adversary strategies, even in a low-data regime. To highlight this advantage, we compare it with a least-squares-based estimation strategy, which can be interpreted as an implementation of the method from \cite{LQ_games_learning}. This approach does not take uncertainty explicitly into account, instead providing a punctual solution based on the available dataset. We show that such an approach may yield unstable behavior under certain adversarial strategies, in contrast with the proposed robust solution.
Specifically, we limit learning to nine data samples and apply the resulting feedback strategies $\hat{K}_1^{robust}$ (robust solution) and $\hat{K}_1^{LS}$ (least-squares estimation solution), evaluating them against extremal unfalsified adversary strategies. Using the same example as before, we present the results in Fig. \ref{rob_comparison}, which depicts the evolution of the state variable $x_1(t)$ under extremal adversary strategies. The right plot highlights a failure case where the least-squares-based strategy leads to instability, while the left plot demonstrates that the robust solution enacted through our method stabilizes the system under all eight extremal adversary strategies.
Through the previous examples, we have demonstrated our method's effectiveness in stabilizing extremal adversary strategies while ensuring convergence to near-optimality given sufficient data.

\begin{figure}[t]
		\includegraphics[clip,trim=1cm 0.4cm 2.9cm 0.88cm,width=0.48\textwidth]
        {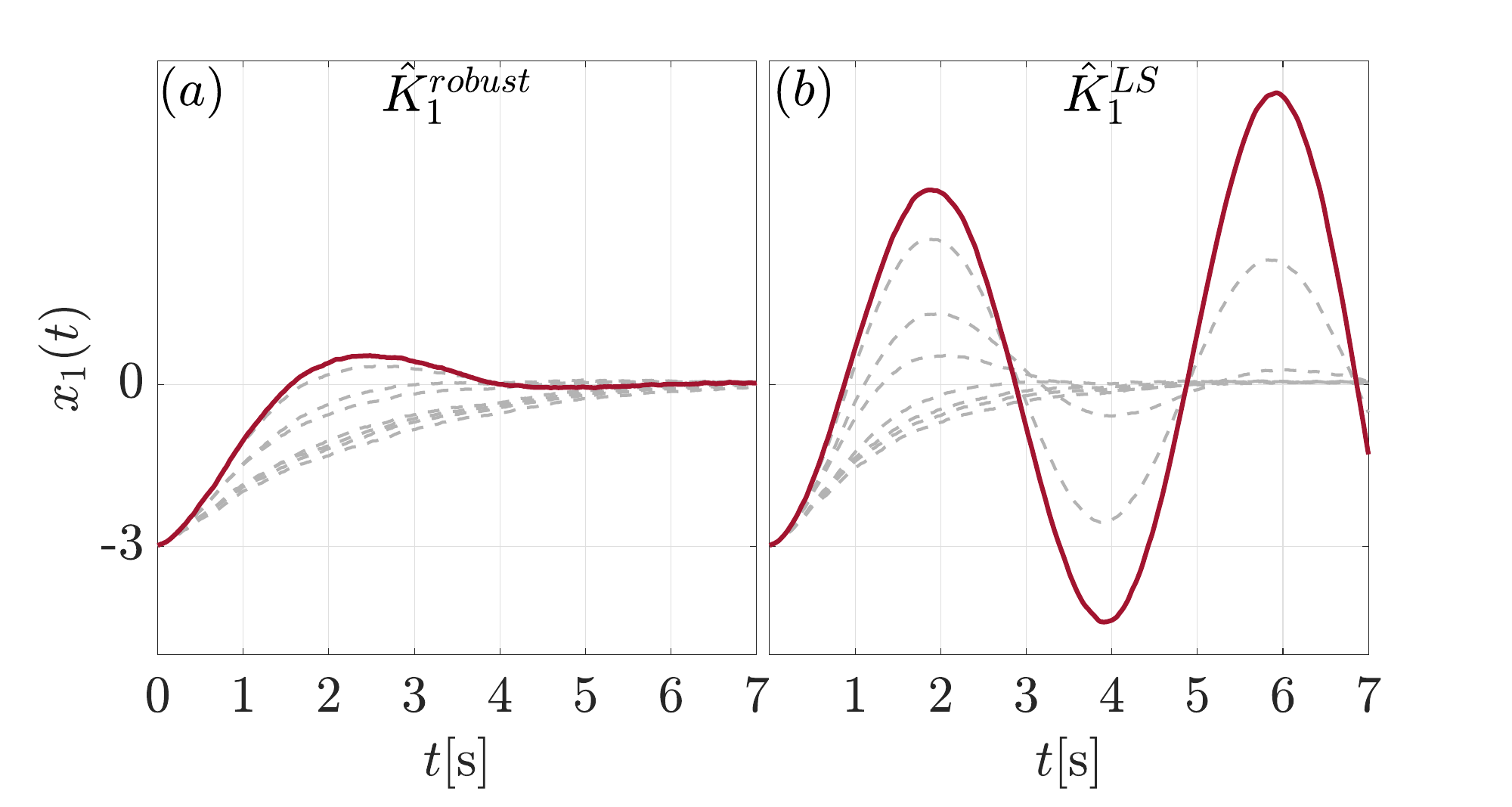}
        \vspace{-0.3cm}
        \caption{Comparison of $x_1(t)$ evolution for extremal adversary strategies. The left panel (a) shows the response under the robust solution $\hat{K}_1^{robust}$, while the right panel (b) shows the response under the least-squares estimation solution $\hat{K}_1^{LS}$. The adversarial strategy for which the least-squares-based method fails to guarantee stability is highlighted in red.}
        \vspace{-0.4cm}
        \label{rob_comparison}
\end{figure}

\subsection{Scalability} \vspace{-0.05cm}
We now demonstrate that the proposed algorithm scales effectively to higher-dimensional problems through the following numerical example. Consider the system dynamics with state and input matrices
\begin{align}
\label{sys2}
\bm{A} &=
\begin{bmatrix}
0 & 1 & 0\\
0 & 0 & 1\\
0.2 & -0.5 & 0.1
\end{bmatrix}, \quad
\bm{B}_1 =
\begin{bmatrix}
0.7 \\
0.7 \\
0.9
\end{bmatrix}, \quad
\bm{B}_2 =
\begin{bmatrix}
0.3 \\
0.4 \\
0.8
\end{bmatrix}.
\end{align}
In this setting, the number of uncertain parameters increases to nine, corresponding to all entries of $\bm{B}_2 \bm{K}_2$, in contrast with only two unknown parameters in the previous example. To represent this uncertainty, each entry of $\bm{B}_2 \bm{K}_2$ 
is assumed to lie within the range $[\bm{0}, 2 \bm{B}_2 \bm{K}_2]$. The system is also subject to an exogenous disturbance, modeled as a uniform noise vector with bounds 
$|w_i| \leq 0.25$. Additionally, the opponent's adaptation dynamics are unknown and set as 
$\tilde{u}(t) = 0.5 \cos(2 \pi t) e^{-0.4 t}$, capturing a decaying oscillatory behavior. Despite the increase in dimensionality, the proposed algorithm remains effective: the controlled agent's policy converges to a neighborhood of the Nash equilibrium within a finite number of iterations, as illustrated in Fig.~\ref{sim2}.

\begin{figure}[t]
		\includegraphics[clip,trim=3.1cm 0.75cm 2.7cm 0.9cm,width=0.48\textwidth]
        {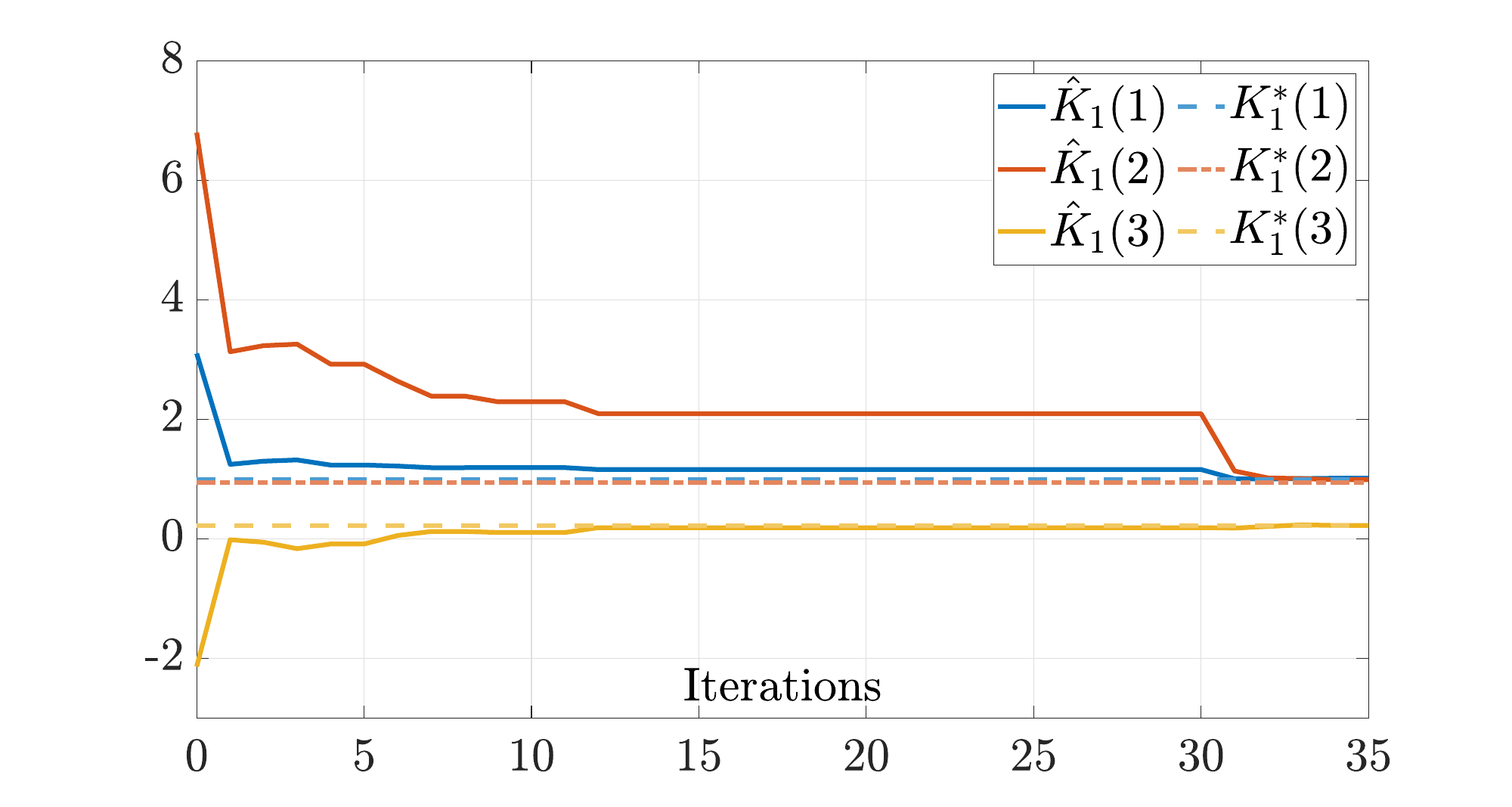}
        \vspace{-0.2cm}
        \caption{The solid lines represent the evolution of the controlled agent's feedback policy, \( \hat{\bm{K}}_1 \), across algorithm iterations, for system (\ref{sys2}). The corresponding optimal Nash gain values are depicted by dashed lines.}
        \vspace{-0.4cm}
        \label{sim2}
\end{figure}

\section{Conclusions} \vspace{-0cm}
This paper proposes a learning-based control approach for linear-quadratic games that ensures robustness against both an unknown adversary strategy and external disturbances. Using set-membership methods, the adversary’s strategy is iteratively estimated while refining the consistent strategy set under bounded uncertainty. To mitigate adversary-induced uncertainty, we incorporate a robust linear quadratic regulator via linear matrix inequalities, ensuring convergence to an $\epsilon$-Nash equilibrium.  
Our framework guarantees robustness and convergence without requiring strong assumptions on adversary adaptation. The method enables online strategy adaptation in competitive, uncertain environments. Numerical simulations demonstrate its effectiveness in interactive decision-making scenarios, including human-robot interaction and multi-agent control. Future work will focus on extending the framework to nonlinear games and exploring probabilistic estimation methods to reduce conservativeness while preserving robustness.

\section*{Acknowledgment} \vspace{-0cm}
This work was supported by the Innovation Network eXprt of the Technical University of Munich, funded by the Federal Ministry of Education and Research (BMBF) and the Free State of Bavaria under the Excellence Strategy of the Federal Government and the Länder, and by the European Research Council (ERC) Consolidator Grant "Safe data-driven control for human-centric systems (CO-MAN)" under grant agreement number 864686. \\ The authors acknowledge the use of ChatGPT \cite{openai2025} for grammar correction and rephrasing of the manuscript text.

\bibliographystyle{unsrt}
\bibliography{ms.bib}

\begin{thebibliography}{10}

\bibitem{isaacs1965differential}
R.~Isaacs.
\newblock {\em {Differential Games: A Mathematical Theory with Applications to Warfare and Pursuit, Control and Optimization}}.
\newblock Dover books on mathematics. Wiley, 1965.

\bibitem{Starr1969}
A.~W. Starr and Y.~C. Ho.
\newblock {Non-Zero-Sum Games Differential Games}.
\newblock {\em Journal of Optimization Theory and Applications}, 3(3):184--206, Mar 1969.

\bibitem{MUSIC202010216}
S.~Musić and S.~Hirche.
\newblock {Haptic Shared Control for Human-Robot Collaboration: A Game-Theoretical Approach}.
\newblock {\em IFAC-PapersOnLine}, 53(2):10216--10222, 2020.
\newblock 21st IFAC World Congress.

\bibitem{Li_HRI}
Y.~Li, K.~P. Tee, R.~Yan, W.~L. Chan, and Y.~Wu.
\newblock {A Framework of Human–Robot Coordination Based on Game Theory and Policy Iteration}.
\newblock {\em IEEE Transactions on Robotics}, 32(6):1408--1418, 2016.

\bibitem{JAS-2019-0028}
X.~Na and D.~J. Cole.
\newblock {Modelling of a Human Driver’s Interaction with Vehicle Automated Steering Using Cooperative Game Theory}.
\newblock {\em IEEE Journal of Automatica Sinica}, 6(JAS-2019-0028):1095, 2019.

\bibitem{basar_book}
T.~Basar and G.~J. Olsder.
\newblock {\em {Dynamic noncooperative game theory. 2nd ed. (Classics in applied mathematics 23)}}.
\newblock Society for Industrial and Applied Mathematics, United States, 1999.

\bibitem{li2019differential}
Y.~Li, G.~Carboni, F.~Gonzalez, D.~Campolo, and E.~Burdet.
\newblock {Differential Game Theory for Versatile Physical Human–Robot Interaction}.
\newblock {\em Nature Machine Intelligence}, 1(1):36--43, 2019.

\bibitem{zhang2021multiagentreinforcementlearningselective}
K.~Zhang, Z.~Yang, and T.~Başar.
\newblock {Multi-Agent Reinforcement Learning: A Selective Overview of Theories and Algorithms}, 2021.

\bibitem{VAMVOUDAKIS20111556}
K.~G. Vamvoudakis and F.~L. Lewis.
\newblock {Multi-player Non-Zero-Sum Games: Online Adaptive Learning Solution of Coupled Hamilton–Jacobi Equations}.
\newblock {\em Automatica}, 47(8):1556--1569, 2011.

\bibitem{engwerda}
J.C. Engwerda.
\newblock {\em {LQ Dynamic Optimization and Differential Games}}.
\newblock John Wiley \& Sons, 2005.

\bibitem{li1995lyapunov}
T.~Y. Li and Z.~Gajic.
\newblock {Lyapunov Iterations for Solving Coupled Algebraic Riccati Equations of Nash Differential Games and Algebraic Riccati Equations of Zero-Sum Games}.
\newblock In {\em New Trends in Dynamic Games and Applications}, pages 333--351. Springer, 1995.

\bibitem{engwerda2007algorithms}
J.~Engwerda.
\newblock {Algorithms for Computing Nash Equilibria in Deterministic LQ Games}.
\newblock {\em Computational Management Science}, 4:113--140, 2007.

\bibitem{yang2019data}
Y.~Yang, S.~Zhang, J.~Dong, and Y.~Yin.
\newblock {Data-Driven Non-Zero-Sum Game for Discrete-Time Systems Using Off-Policy Reinforcement Learning}.
\newblock {\em IEEE Access}, 8:14074--14088, 2019.

\bibitem{LQ_games_learning}
B.~Nortmann, A.~Monti, M.~Sassano, and T.~Mylvaganam.
\newblock {Nash Equilibria for Linear Quadratic Discrete-Time Dynamic Games via Iterative and Data-Driven Algorithms}.
\newblock {\em IEEE Transactions on Automatic Control}, 69(10):6561--6575, 2024.

\bibitem{Zhang_online_ADP}
Huaguang Zhang, Lili Cui, and Yanhong Luo.
\newblock {Near-Optimal Control for Nonzero-Sum Differential Games of Continuous-Time Nonlinear Systems Using Single-Network ADP}.
\newblock {\em IEEE Transactions on Cybernetics}, 43(1):206--216, 2013.

\bibitem{ASL2024105936}
Hamed Jabbari~Asl and Eiji Uchibe.
\newblock {Inverse Reinforcement Learning Methods for Linear Differential Games}.
\newblock {\em Systems \& Control Letters}, 193:105936, 2024.

\bibitem{Bisoffi2020ControllerDF}
A.~Bisoffi, C.~De~Persis, and P.~Tesi.
\newblock {Controller Design for Robust Invariance From Noisy Data}.
\newblock {\em IEEE Transactions on Automatic Control}, 68:636--643, 2020.

\bibitem{kerz2024safe}
S.~Kerz, A.~Lederer, M.~Leibold, and D.~Wollherr.
\newblock {Safe Online Nonstochastic Control from Data}.
\newblock In {\em ICML 2024 Workshop: Foundations of Reinforcement Learning and Control -- Connections and Perspectives}, 2024.

\bibitem{ziegler2012lectures}
G.~M. Ziegler.
\newblock {\em {Lectures on Polytopes}}, volume 152.
\newblock Springer Science \& Business Media, 2012.

\bibitem{matousek2013lectures}
J.~Matousek.
\newblock {\em Lectures on Discrete Geometry}, volume 212.
\newblock Springer Science \& Business Media, 2013.

\bibitem{feron_lmi}
E.~Feron, V.~Balakrishnan, S.~Boyd, and L.~El~Ghaoui.
\newblock {Numerical Methods for H2 Related Problems}.
\newblock In {\em 1992 American Control Conference}, pages 2921--2922, 1992.

\bibitem{LMI_LQR}
C.~Olalla, R.~Leyva, A.~El~Aroudi, and I.~Queinnec.
\newblock {Robust LQR Control for PWM Converters: An LMI Approach}.
\newblock {\em IEEE Transactions on Industrial Electronics}, 56(7):2548--2558, 2009.

\bibitem{DataDriven}
C.~De~Persis and P.~Tesi.
\newblock {Formulas for Data-Driven Control: Stabilization, Optimality, and Robustness}.
\newblock {\em IEEE Transactions on Automatic Control}, 65(3):909--924, 2020.

\bibitem{LMI94}
S.~Boyd, L.~El~Ghaoui, E.~Feron, and V.~Balakrishnan.
\newblock {\em {Linear Matrix Inequalities in System and Control Theory}}.
\newblock SIAM studies in applied mathematics: 15, 1994.

\bibitem{MILANESE198678}
M.~Milanese, R.~Tempo, and A.~Vicino.
\newblock Strongly optimal algorithms and optimal information in estimation problems.
\newblock {\em Journal of Complexity}, 2(1):78--94, 1986.

\bibitem{MILANESE1991997}
M.~Milanese and A.~Vicino.
\newblock {Optimal Estimation Theory for Dynamic Systems with Set Membership Uncertainty: An Overview}.
\newblock {\em Automatica}, 27(6):997--1009, 1991.

\bibitem{FOGEL1982229}
E.~Fogel and Y.F. Huang.
\newblock {On the Value of Information in System Identification:Bounded Noise Case}.
\newblock {\em Automatica}, 18(2):229--238, 1982.

\bibitem{Dasgupta}
S.~Dasgupta and Y.-F. Huang.
\newblock {Asymptotically Convergent Modified Recursive Least-Squares with Data-Dependent Updating and Forgetting Factor for Systems with Bounded Noise}.
\newblock {\em IEEE Transactions on Information Theory}, 33(3):383--392, 1987.

\bibitem{sun}
J.-G. Sun.
\newblock {Perturbation Theory for Algebraic Riccati Equations}.
\newblock {\em SIAM Journal on Matrix Analysis and Applications}, 19(1):39--65, 1998.

\bibitem{openai2025}
OpenAI.
\newblock {ChatGPT}.
\newblock \url{https://chat.openai.com}, 2025.

\end{thebibliography}

\end{document}